\documentclass[sigconf, natbib, screen=true, anonymous=false, review=false]{acmart}

\acmSubmissionID{32}

\usepackage{amsmath}
\usepackage{amsfonts}
\usepackage{dsfont}
\usepackage{mathtools}
\usepackage{graphicx}
\usepackage{booktabs} 
\usepackage{numprint}
\npthousandsep{,}
\npdecimalsign{.}
\usepackage[inline]{enumitem}
\usepackage{balance}
\usepackage[skip=0pt]{caption}
\usepackage{multirow}
\usepackage{siunitx}
\usepackage{acronym}
\usepackage{xcolor}
\usepackage{algorithm}
\usepackage{algpseudocode}
\usepackage{caption}
\usepackage{subcaption}
\usepackage{tikz}

\usepackage{csvsimple}

\usepackage{pgfplots}
\usepackage{tikz}
\usetikzlibrary{calc}

\usepackage{xspace}

\AtBeginDocument{%
  \providecommand\BibTeX{{%
    \normalfont B\kern-0.5em{\scshape i\kern-0.25em b}\kern-0.8em\TeX}}}

\newcommand{\Ptil}{\widetilde{P}}

\newtheorem{corollary }{Corollary}

\definecolor{darkgreen}{rgb}{0.0, 0.2, 0.13}
\definecolor{asparagus}{rgb}{0.53, 0.66, 0.42}
\definecolor{cambridgeblue}{rgb}{0.64, 0.76, 0.68}

\newcommand{\maria}[1]{\textcolor{magenta}{#1}}

\newcommand{\rank}{\operatorname{rank}}

\acrodef{LTR}{learning to rank}
\acrodef{MRP}{marginal rank probability\xspace}

\newcommand{\method}{FELIX\xspace}
\newcommand{\methodn}[1]{FELIX$_{iter=#1}$\xspace}

\newcommand{\bvn}{Birkhoff-von Neumann\xspace}

\newcommand{\headernodot}[1]{\vspace{1mm}\noindent\textbf{#1}}
\newcommand{\header}[1]{\headernodot{#1.}}

\acrodef{IR}{information retrieval}
\acrodef{CLTR}{counterfactual learning to rank}
\acrodef{RBP}{ranked-based precision}
\acrodef{ERR}{expected-reciprocal rank}
\acrodef{PBM}{position-based user model}
\acrodef{PL}{Plackett-Luce}
\acrodef{BvN}{Birkhoff-von Neumann}

\allowdisplaybreaks

\copyrightyear{2022}
\acmYear{2022}
\setcopyright{acmlicensed}\acmConference[SIGIR '22]{Proceedings of the 45th International ACM SIGIR Conference on Research and Development in Information Retrieval}{July 11--15, 2022}{Madrid, Spain}
\acmBooktitle{Proceedings of the 45th International ACM SIGIR Conference on Research and Development in Information Retrieval (SIGIR '22), July 11--15, 2022, Madrid, Spain}
\acmPrice{15.00}
\acmDOI{10.1145/3477495.3531977}
\acmISBN{978-1-4503-8732-3/22/07}

\author{%
Maria Heuss
}
\affiliation{%
  \institution{
  University of Amsterdam
  \city{Amsterdam}
  \country{The Netherlands}%
  }
}  
\email{m.c.heuss@uva.nl}

\author{%
Fatemeh Sarvi
}
\affiliation{%
  \institution{
  AIRLab, University of Amsterdam
  \city{Amsterdam}
  \country{The Netherlands}%
  }
}  
\email{f.sarvi@uva.nl}

\author{%
Maarten de Rijke%
}
\affiliation{%
  \institution{
  University of Amsterdam
  \city{Amsterdam}
  \country{The Netherlands}%
  }  
}  
\email{m.derijke@uva.nl}

\begin{document}

\title{Fairness of Exposure in Light of Incomplete Exposure Estimation}

\renewcommand{\shortauthors}{Heuss, et al.}

\begin{abstract}
Fairness of exposure is a commonly used notion of fairness for ranking systems. It is based on the idea that all items or item groups should get exposure  proportional to the merit of the item or the collective merit of the items in the group. 
Often, stochastic ranking policies are used to ensure fairness of exposure. 
Previous work unrealistically assumes that we can reliably estimate the expected exposure for all items in each ranking produced by the stochastic policy.
In this work, we discuss how to approach fairness of exposure in cases where the policy contains rankings of which, due to inter-item dependencies, we cannot reliably estimate the exposure distribution. In such cases, we cannot determine whether the policy can be considered fair. 
Our contributions in this paper are twofold.
First, we define a method called \method{} for finding stochastic policies that avoid showing rankings with unknown exposure distribution to the user without having to compromise user utility or item fairness. 
Second, we extend the study of fairness of exposure to the top-$k$ setting and also assess \method{} in this setting.
We find that \method{} can significantly reduce the number of rankings with unknown exposure distribution without a drop in user utility or fairness compared to existing fair ranking methods, both for full-length  and  top-$k$ rankings. 
This is an important first step in developing fair ranking methods for cases where we have incomplete knowledge about the user's behaviour. 
\end{abstract}



\maketitle

\acresetall


\section{Introduction}

There has been increased interest in fair ranking systems, as witnessed by the number of publications~\citep{ekstrand-2021-fairness,zehlike2021fairness}, the topic's attention during keynotes leading conferences~\citep{castillo2019fairness, joachims2021fairness}, and challenges such as the TREC Fair Ranking track \cite{trec-fair-ranking-2021}.
Several particularities about rankings make this task especially challenging. 

\begin{figure}\label{fig:visualization_felix}
     \centering
     \begin{subfigure}[b]{0.22\textwidth}
         \centering
         \includegraphics[width=\textwidth]{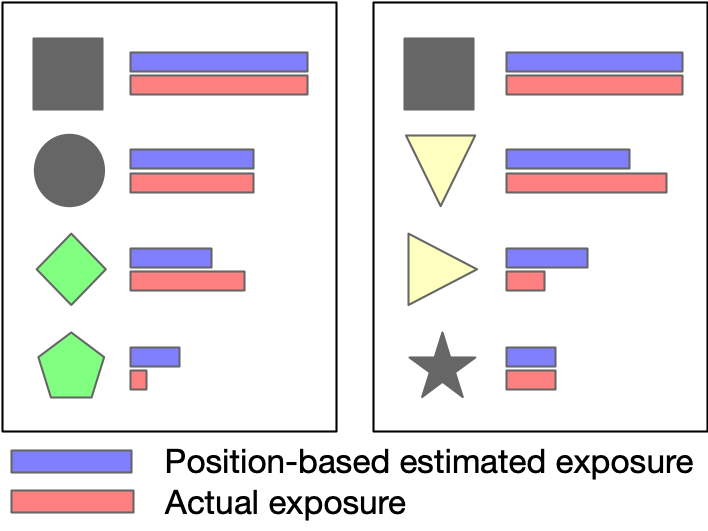}
         \caption{}
         \label{fig:unknown_exposure}
     \end{subfigure}
     \hfill
     \begin{subfigure}[b]{0.22\textwidth}
         \centering
         \includegraphics[width=\textwidth]{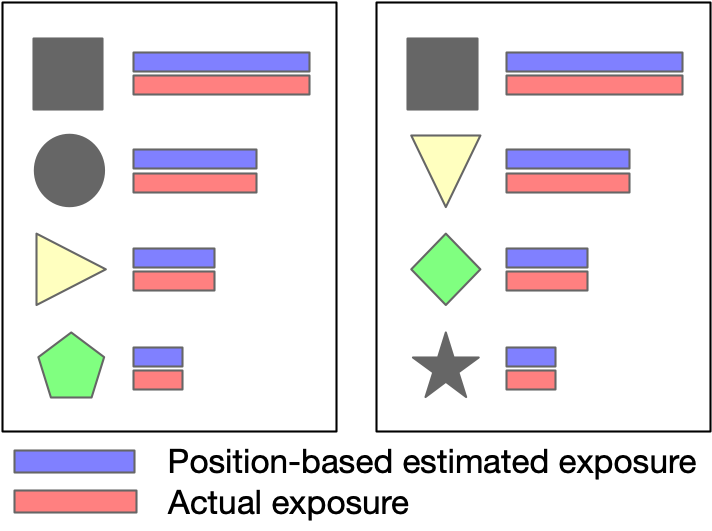}
         \caption{}
         \label{fig:known_exposure}
     \end{subfigure}
        \caption{
        Visualization of rankings with unknown exposure distribution which are due to inter-item dependencies between items marked by the same color and similar shapes (a).  By shuffling some items between rankings in the stochastic ranking policy these dependencies can be reduced such that the estimated exposure agrees with the actual exposure that each item gets (b).
}
\end{figure}

First, often ranking systems act as a tool for two-sided marketplaces,  such as job markets \cite{geyik2019fairness} or music recommender systems~\citep{mehrotra2018towards}. On one side, users want relevant item recommendations. On the other side, items or their providers are interested in being exposed to as many users as possible.
Second, biases like position bias can cause a traditional deterministic ranking 
to amplify small differences in predicted scores into vast differences in user attention~\citep{singh2018fairness, biega2018equity}. 

An important line of research on fairness in ranking deals with \emph{fairness of exposure}. 
Given a ranking, we can estimate how much exposure each item gets in expectation during inference. 
We call this the \emph{exposure distribution} of the ranking. 
\citet{singh2018fairness} define several notions of fairness of exposure for rankings, among them \emph{disparate treatment}. This notion defines a stochastic ranking policy to be fair if each item or item-group gets expected exposure proportional to its merit. 
We will mostly focus on individual fairness, where we want to provide each individual item with exposure relative to its merit. 

\header{Incomplete exposure estimation}
Previous methods for fairness of exposure assume that we can estimate the exposure distribution of any ranking in the set of all possible rankings. 
For this, a user model like the position-based model \cite{singh2018fairness, wang2020fairness,biega2018equity, yadav2021policy},
or the ERR-based model \cite{diaz2020evaluating}
can be used. 
However, there are cases where, due to inter-item dependencies
that are not accounted for by any of the existing user models,
for certain rankings, user-behaviour does not follow the user model; for such rankings we cannot estimate the exposure distribution accurately.  
See Fig.~\ref{fig:unknown_exposure} for an illustration.
E.g., \citet{sarvi-2021-understanding-arxiv} show that visual outliers can have a great impact on the exposure distribution within a ranking, since such outliers attract more user attention. 
This phenomenon is an example of inter-item dependencies where one item can be perceived as an outlier in the context of items it is presented together with. It can cause the exposure distribution to diverge from the distribution assumed by the user model.

Simply ignoring the incomplete knowledge about the exposure of some of the rankings would imply that we cannot guarantee fairness. 
Also, by ignoring potentially incomplete exposure estimation, we might introduce a new kind of bias into the collected click data, since items that got more exposure than estimated will have propensity values that are too high, leading to overestimation of their relevance.  
One solution would be to obtain
a more accurate user browsing model by estimating the exposure distribution of rankings that do not follow the user model, through a large-scale user study. 
To the best of our knowledge no such studies have been conducted.
It is also not clear whether one can always reliably estimate the exposure distribution for all possible rankings.  

Instead, we propose to avoid showing rankings with unknown exposure distribution to the user by reducing their weight in the probability distribution of the stochastic ranking policy. 

\header{Fair top-$k$ ranking} So far, the literature on fairness of exposure has mostly focused on full-length rankings. 
Top-$k$ rankings are well studied in the general \ac{IR} literature~\citep{deshpande-2008-efficient,chen-2015-top-k,zehlike2017fa,yang-2012-top-k}; many real-world ranking applications require us to expose just a short list of items.
Often there are more relevant items than can be shown to the user,  hence it is important to consider fairness of exposure for this set-up as well.  
Although there have been few approaches to \emph{fair top-$k$ ranking}~\citep{zehlike2022fair,zehlike2017fa},  most are concerned with demographic parity, rather than merit-based fairness of exposure.  

\header{Our contributions}
In this work we develop a method to find ranking policies that avoid presenting rankings with unknown exposure distribution, while still optimizing for user utility and fairness. Under the assumption that inter-item dependencies are the reason for the shift in exposure, our method works by shuffling items between different rankings to avoid presenting them in a context where they disturb the position-based exposure distribution, as illustrated in Fig.~\ref{fig:known_exposure}.

We also present what we believe to be the first approach towards fairness of exposure in the top-$k$ setting for the convex optimization approach towards fairness. 
We generalize the \bvn theorem and use this to extend~\citep{singh2018fairness} to the top-$k$ setting. 
 
To summarize, our main contributions are as follows: 
\begin{itemize}[leftmargin=*,nosep]
\item We introduce the task of fairness of exposure in light of incom\-plete exposure estimation and define a novel method \method that provides us with a fair ranking policy that avoids rankings with unknown exposure distribution. 
\item To make \method{} applicable to a broader range of use cases, we extend the constrained optimization approach to fairness of exposure to the top-$k$ case. 
\item We test and compare \method{} on the outlier use case introduced in \cite{sarvi-2021-understanding-arxiv} and show big improvements over other top-$k$ fair ranking methods in terms of effectiveness in avoiding rankings containing outliers, while staying within the fairness constraints. 
\end{itemize}


\section{Related Work}

\textbf{Fairness in ranking.}
For a detailed overview of fair ranking we refer to \cite{ekstrand-2021-fairness,zehlike2021fairness}. 
\citet{yang2017measuring} seem to have been the first to formalize fairness for rankings in a rank-aware manner, by calculating parity for different top-$k$ cut-offs and summing over these values with a rank-based discount. 
\citet{zehlike2017fa,  zehlike2022fair} discuss representational fairness for top-$k$ rankings and define a re-ranking algorithm that ensures a share of items from the protected groups in every prefix of the top-$k$, while
\citet{celis2017ranking} formulate the problem as a constrained optimization problem.  
These papers look for a deterministic ranker, not a stochastic ranking policy, and emphasize on representational fairness and demographic parity. 

\citet{singh2017equality} introduce the notion of expected exposure and define fairness of exposure with respect to demographic parity and equal opportunity, where the expected exposure is calculated w.r.t.\ position bias.  
Later work~\citep{singh2018fairness} defines different types of fairness of exposure w.r.t.\ disparate impact and disparate treatment, and address the task as a constrained optimization problem.
\citet{biega2018equity} define equity of attention as an alternative notion of fairness for rankings that is also based on exposure; they also address the task as a constrained optimization problem. 
\citet{wang2020fairness} also consider fairness of exposure  combined with diversity in rankings. 
We build on \citep{singh2018fairness} and use the non-uniqueness property of the Birkhoff-von Neumann decomposition that is also used in \citep{wang2020fairness} to produce more diverse rankings.  
Importantly, we reduce the probability that the user is shown a ranking with unknown exposure distribution rather than providing the user with more diverse rankings as in~\citep{wang2020fairness}. 

Another line of research aims to include fairness in the learning process by including a fairness objective in the objective function \citep{singh2019policy, zehlike2020reducing, diaz2020evaluating,vardasbi-2022-fairness}.
Since inter-item relationships are hard to model within the in-processing set-up, in our work we focus on a post-processing method for avoiding rankings with unknown exposure distribution and leave work on in-processing methods for the future. 

Another work that looks into the the topic of uncertainty within fair ranking is \cite{singh2021fairness},  which explores fairness of exposure when there is uncertainty about the merit.  In contrast to this work, we are considering uncertainty about the exposure of certain rankings. 

\header{Exposure estimation in ranking}
In \ac{CLTR} true estimation of exposure plays a central role~\citep{joachims2017unbiased}. 
Early work on CLTR corrects for position bias using exposure, estimated by a click model~\citep{chuklin-2015-click}, as the propensity to inversely weight the importance of clicks~\cite{joachims2017unbiased, wang2016learning}. 
More recent work focuses on estimating examination probabilities \cite{agarwal2019estimating, ai2018unbiased, fang2019intervention, wang2018position,vardasbi-2020-inverse,vardasbi-2021-mixture-based},  which also correlates with exposure, correcting for more types of bias.
Recent work on learning fair rankings from implicit feedback \citep{yadav2021policy} simultaneously corrects for position bias and implicit biases in the data. 
There is no prior work on how to adapt these models for the case where certain rankings do not follow the general user model. 

Prior work has shown that exposure might be impacted by other factors than just position and the relevance of other items.  
\citet{yue2010beyond} observe that visual attractiveness can impact the exposure that items get; \citet{sapiezynski2019quantifying} acknowledge that the attention that users give to items in a ranking depends on context; and \citet{wang2021clicks} address the impact of click bait items on exposure distribution.  
\citet{sarvi-2021-understanding-arxiv} show that the existence of visual outliers in rankings can skew the exposure distribution amongst the items, causing outliers to draw more attention than estimated by the position-based user model that non-outlier rankings seem to follow.

In this work, we focus on similar but more general use cases, where due to inter-item relationships the exposure distribution for some rankings differs from the generally assumed distribution,  that can be described through existing user models. 


\section{Background}
\label{sec:background}
We introduce preliminaries in fair ranking that form the basis for a new method for ranking under fairness constraints, while avoiding to present rankings with unknown exposure distribution.

\subsection{Stochastic ranking policies}
\label{sec:background_prob_pol}
Depending on the definition of fairness being used,  often a single deterministic ranking cannot achieve fairness~\citep{diaz2020evaluating, biega2018equity}.
Instead, probabilistic rankers can be used to provide a fair distribution of exposure among items.  
Given a query $q$ and set of candidate items, $\mathcal{D}_q = \{d_i\}_{i=1,\dots, n}$,  to be ranked, we define a \emph{stochastic ranking policy} $\pi_q$ as a probability distribution over all possible rankings 
$\mathcal{R}_{\mathcal{D}_q}$.  
That is,  $\pi$ assigns each ranking $\sigma_j\in \mathcal{R}_{\mathcal{D}_q}$ a probability $\pi_q(\sigma_j)$ that it will be shown to the user. 

To evaluate the fairness of a ranking policy we determine the \emph{expected exposure} $\epsilon(d_i \mid \pi_q)$ that each item $d_i$ \if0\maria{Removed item group}\fi obtains when enough rankings have been presented to users.  
To compute this, we need to assume a browsing model that explains the probability of a user visiting an item. \citet{diaz2020evaluating} adopt user models corresponding to the \ac{RBP} and \ac{ERR}, while \citet{singh2018fairness} use the \ac{PBM}.
We follow the latter, as it is commonly used in the fairness literature \citep{singh2018fairness,wang2020fairness,biega2018equity,yadav2021policy}. 
Assuming that the exposure of an item in a ranking,  $\epsilon(d_i \mid \sigma)$,  is purely based on its position, the \emph{expected exposure} $\epsilon(d_i \mid \pi_q)$ of document $d_i$ for policy $\pi_q$ can be calculated as: 
\begin{equation}
\begin{split} 
\epsilon(d_i \mid \pi_q)&= \mathbb{E}_{\sigma \sim \pi_q}   \epsilon(d_i \mid \sigma) \\
 &= \sum_{\sigma \in \mathcal{R}_{D_q}} \pi_q(\sigma) \cdot \epsilon(d_i \mid \sigma) \\ 
&= \sum_{\sigma \in \mathcal{R}_{D_q}} \pi_q(\sigma) \cdot \frac{1}{\log(1+\rank(d_i \mid \sigma))},
\end{split}
\end{equation}
where we assume that the exposure can be calculated based on the rank:
$ \epsilon(d_i \mid \sigma) = v (rank(d_i \mid \sigma))$ with exposure at rank $j$ given by $v(j) = \frac{1}{\log(1+j)}$.

\subsection{Fairness of exposure}
\label{sec:background_foe}
The definition of what constitutes a fair ranking may vary between application scenarios and types of biases being addressed \cite{zehlike2021fairness}.  
We focus on individual fairness, but our approach can easily be extended for group fairness. 
Our goal is to make sure that similar items receive a similar amount of exposure that is proportional to their merit. 
The \emph{merit} $u(d\mid q)$ of an item, $d\in \mathcal{D}$, indicates how much exposure it deserves to get from users with respect to query $q$. 
We define the merit of an item as its relevance to the query.  

The idea of \emph{fairness of exposure} \cite{singh2018fairness} is to provide each item with exposure $\epsilon$ that is proportional to its merit:
\begin{equation}
\frac{\epsilon(d_i \mid \pi_q)}{u(d_i \mid q)} = \frac{\epsilon(d_j \mid \pi_q)}{u(d_j\mid q)} \hspace{0.5cm} \forall d_i, d_j\in \mathcal{D}.
\label{equation:disparity}
\end{equation}

\subsection{Finding a stochastic policy under fairness constraints}\label{sec:constraint_optimization}
To be able to satisfy certain fairness constraints, we need to find a stochastic ranking policy (Section~\ref{sec:background_prob_pol}). 
\citet{singh2018fairness} approach the problem by optimizing for user utility under fairness constraints via linear programming.
As our method is based on theirs, we introduce it in more detail.  
For each query $q$ and item $d\in\mathcal{D}$, let 
$u(d \mid q)$ be its relevance to the user.
We define the \emph{utility} $U$ of a ranking policy $\pi_q$ as the expected utility to the user, when shown a ranking sampled from $\pi_q$: 
\begin{equation}
\begin{split}
U(\pi_q) 	&= \sum_{d \in \mathcal{D}} \epsilon(d \mid \pi_q) \cdot u(d \mid q) \\
& =  \mathbb{E}_{\sigma\sim \pi_q} \sum_{d \in \mathcal{D}} \epsilon(d \mid \sigma) \cdot  u(d \mid q).
\end{split}
\label{eq:utility}
\end{equation}
As we assume a position-based user model, $\epsilon(d\mid \sigma)$ is purely dependent on the position of $d$ in the ranking. 
Therefore, the expected utility $U$ can be calculated  based on the probabilities $P_{i,j}=P(d_i \text{ is placed at rank } j)$: 
\begin{equation}
\begin{split}
U(\pi_q) 	&=  \sum_{d_i \in \mathcal{D}} \sum_{j \in \{1,\ldots, n\}} P_{i.j} \cdot v(j)  \cdot u(d_i\mid q)\\ 
	&= \mathbf{u}^{T} \mathbf{P}  \mathbf{v},
\end{split}	
\label{eq:utility:rephrased}
\end{equation}
where $n=|\mathcal{D}|$ is the number of items in the ranking, $\mathbf{u}$ the vector containing the merit of each item,  $\mathbf{v}$ the vector containing the position bias at each position, and $\mathbf{P}=\{P_{i,j}\}_{i,j = 1,\dots, n}$.  
\citet{singh2018fairness} show that the disparate treatment constraint from Eq.~\eqref{equation:disparity} can be formulated as a linear constraint in $\mathbf{P}$, which yields a convex optimization problem of the form:
\begin{align} 
\begin{split}
\mathbf{P}	=  \operatorname{argmax}_{\mathbf{P}}  & ~\mathbf{u}^{T} \mathbf{P} \mathbf{v} \\
\text{such that }		&\mathds{1}^{T} \mathbf{P}=\mathds{1} \\
					&\mathbf{P} \mathds{1}=\mathds{1}  \\
					&0 \leq P_{i, j} \leq 1 \\
					&\mathbf{P} \text{ is fair}.
\end{split}
\label{eq:originaloptimization}
\end{align} 
A solution $\mathbf{P}$ to this optimization problem is a doubly stochastic matrix, called the \acfi{MRP} \emph{matrix}.
The solution $\mathbf{P}$ needs to be transformed into an executable stochastic ranking policy.  
The Birkhoff-von Neumann theorem \citep{birkhoff1946tres} gives us a constructive proof that such a matrix can be decomposed into a convex sum of $M\le n^2-n+1$ permutation matrices: 
\begin{equation}
\mathbf{P} = \!\sum_{m=1,\dots, M} \!\alpha_m P_{\sigma_m} \text{ such that} \sum_{m=1,\dots, M} \!\alpha_m = 1 \, (0\le \alpha_m \le 1). \label{eq:original_bvn}
\end{equation}
Since each permutation matrix corresponds to some ranking, we denote the permutation matrix corresponding to $\sigma$ by $P_\sigma$.  
 
With this we have found a stochastic policy $\pi$ with $\pi(\sigma_m) = \alpha_m$ and  $\pi(\sigma) = 0$ for all $\sigma$ not contained in this convex sum.  
Note that this decomposition is not necessarily unique; in Section~\ref{sec:method_felix} below we will make use of this fact.

\subsection{The impact of outliers on the exposure in rankings}\label{sec:background_outliers}

\citet{sarvi-2021-understanding-arxiv} provide evidence that commonly made assumptions on the user-behaviour might not hold when the presented ranking contains visible outliers that might attract the attention of the user.  
Since outliers are an example where inter-item dependencies between documents can change the exposure distribution among the items in a ranked list, we work with this example for our experiments in Section~\ref{sec:experiments}. 
We follow the set-up of \citep{sarvi-2021-understanding-arxiv}, where the authors assume that outliers can be determined through outlier detection on a specific visual item feature $ g(d)$ that might impact the user's perception of an item. 
In the case of scholarly search, which is used as an example in the experiments, such a feature could be the number of citations that each document has.  

Outliers are considered in a context $C\subset \mathcal{D}$ of  items that are presented together, which could for instance be the top-$k$ that is presented in a single search engine result page (SERP).
Given such a context $C = \{d_1, \dots, d_k\} \subset \mathcal{D}$,  we use  the features, $ g(d_1), \dots, g(d_k)$,  as input for the outlier detection.  
\citet{sarvi-2021-understanding-arxiv} find that the performance of their method for removing outliers from the rankings is not very sensitive to the outlier detection method.
For simplicity, we will therefore use the Z-score:
\begin{equation}
z(g_i) = \frac{g_i - \mu}{s},
\end{equation}
where $g_i = g(d_i)$,  and $\mu=\frac{1}{k}\sum_{i=1}^k g_i$ and $s=\sqrt{\frac{1}{k}\sum_{i=1}^k (g_i-\mu)^2}$ denote the mean and standard deviation of the scores in that context.  
Given these Z-scores, we define an item $d_i$ to be an outlier if $|z(g_i)| > \lambda$, where $\lambda$ can be chosen dependent on the sensitivity towards outlier items. 
Here, we diverge slightly from \citep{sarvi-2021-understanding-arxiv}, who use a more complex outlier detection method. 

\medskip\noindent%
Next, we introduce an extension to the convex optimization approach to fairness of exposure from Section~\ref{sec:constraint_optimization} for top-$k$ rankings.  
We use the definition of fairness of exposure with respect to disparate treatment from Section~\ref{sec:background_foe} and work with stochastic policies from Section~\ref{sec:background_prob_pol}.  
We also develop a method that avoids displaying rankings with unknown exposure distribution, using the outlier use case from Section~\ref{sec:background_outliers} for our experiments in Section~\ref{sec:experiments}. 


\section{Fairness of exposure under incomplete exposure estimation}
\label{sec:method}

As discussed in Section~\ref{sec:background_prob_pol}, previous work on fair ranking assumes that we can estimate the exposure distribution for all rankings in a policy with one user model. Often, the position-based user model is used.  But there are cases where these assumptions do not hold up. 
\citet{sarvi-2021-understanding-arxiv} show that the existence of outliers in a displayed ranking can strongly impact the exposure distribution of the ranking.  
To the best of our knowledge, there is no prior work on estimating the exposure distribution of such rankings.  
If such \emph{rankings with unknown exposure distribution} are part of a stochastic ranking policy (i.e., if such a ranking has a non-zero probability of being presented to the user), we cannot determine whether the policy is fair.  Therefore, for attaining fair stochastic policies we should avoid using such rankings.  This introduces the task of fair ranking under incomplete exposure estimation. 

In this section we develop a method for the task of \textbf{F}airness of \textbf{E}xposure in \textbf{L}ight of \textbf{I}ncomplete e\textbf{X}posure estimation, \method{}, that provides a ranking policy that avoids rankings with unknown exposure distribution without damaging fairness or utility.   
\method{} is based on the assumption that the shift in the exposure distribution is caused by inter-item relationships  between the items that are ranked together.  
Hence, depending on the context an item is presented in, it could either follow the position-based exposure distribution or it could draw more or less exposure than assumed.  In the example,  an outlier in a ranking might draw more attention than a non-outlier item at the same position,  as demonstrated in~\citep{sarvi-2021-understanding-arxiv}.  When presented in a more diverse ranking, the same item might not be considered an outlier any more and follow the assumed position-based exposure distribution.   
Compared to the method for removing outliers from the top-$k$ in \citep{sarvi-2021-understanding-arxiv}, \method{} is more generally applicable to any use case where, due to inter-item dependencies, some rankings have unknown exposure distribution.  
Also, \method{} allows us to consider outliers in the local context that they are presented in, while \citeauthor{sarvi-2021-understanding-arxiv}'s approach can only remove outliers with respect to the global context of all items in the list.

Since the context in which items are presented in plays a central role for our task,  naturally we are interested in our method to work in the top-$k$ setting.  Therefore, we first generalize the constrained optimization approach towards fairness of exposure, introduced in~\cite{singh2018fairness}, to the top-$k$ setting and present an efficient way to determine a fair policy. 
Then we present our method \method{} that uses iterative re-sampling to determine a stochastic policy that avoids presenting rankings with unknown exposure distribution to the user, while staying within the fairness constraints. 

\subsection{Fair ranking in the top-$k$ setting}
\label{sec:method_topk}

We will now extend the convex optimization approach to fairness to the top-$k$ setting.   
Let $n$ be the number of candidate items to be ranked and  $k\le n$ be the number of ranks of the desired rankings. 
As explained in Section~\ref{sec:constraint_optimization}, searching for a stochastic policy under fairness constraints can be done by first searching for a marginal rank probability matrix $\mathbf{P}$ that satisfies the fairness constraints, and then decomposing this matrix. 
Since we are interested in the top-$k$ case, $\mathbf{P}=\{P_{i,j}\}_{i=1,\dots n,j=1,\dots k}$ is now a $n\times k$ matrix,  where $P_{i,j}$ is the probability that item $i$ is placed at rank $j$.   
With $\mathbf{u}$ the $n$-dimensional utility vector and $\mathbf{v}$ the $k$-dimensional vector containing the examination probability at each of the top-$k$ positions we can solve the following linear program:
\begin{equation}
\label{eq:convex_optimization_topk}
\begin{split}
\mathbf{P}	=  \operatorname{argmax}_{\mathbf{P}} &~\mathbf{u}^{T} \mathbf{P} \mathbf{v}  \\
\text{such that }	& \mathds{1}_n^T\mathbf{P} = \mathds{1}_k  \\ 
	& \mathbf{P}\mathds{1}_k\le\mathds{1}_n  \text{ (element-wise inequality)}  \\
					&0 \leq P_{i, j} \leq 1  \\ 
					&\mathbf{P} \text{ is fair}.
\end{split}
\end{equation}
Given the \acl{MRP} matrix $\mathbf{P}$,  we want to determine a stochastic policy given by a distribution over actual rankings. In the $n\times n$ setting, the \acf{BvN} decomposition provides us with an algorithm to determine such a distribution.  The following result generalizes the \acs{BvN} theorem to the $n\times k$ setting where $n$ is not necessarily equal to $k$. 

\begin{theorem}
\label{theorem:generalized_bvn}
Any matrix $P=\{a_{i,j}\}_{i\le n,j\le k}$  
with $\forall i,j: 0 \le a_{i,j} \le 1$,  $\forall j: \sum_{i=1}^n a_{i,j} = 1$ and   $\forall i: \sum_{j=1}^k a_{i,j} \le 1$ can be written as the convex sum $P=\sum_{l=1}^m \alpha_l \cdot P_l$ of permutation matrices $P_l$ with coefficients $\alpha_l \in [0,1]$ such that $\sum_{l=1}^m \alpha_l=1$.  
\end{theorem}
\begin{proof}
 In Lemma \ref{theorem:extension_doubly_stochastic} below, we show that  $P$ can be extended to a doubly stochastic matrix $P'$. We can use the \acs{BvN} decomposition for doubly stochastic matrices to find a decomposition for $P'$, which will induce a decomposition for $P$.  For details, see the Appendix.
\end{proof}

\noindent%
Here we say that $P'\in \mathbb{R}^{n'\times k'}$ is an \emph{extension} of $P\in \mathbb{R}^{n\times k}$ if $n'\ge n, k' \ge k$, and $P_{i,j}=P'_{i,j}$ for all $(i,j)$ with $i\le n$ and $j\le k$. We will denote this by $P'\vert_{i\le n, j\le k} = P$. 

\begin{lemma}
\label{theorem:extension_doubly_stochastic}
Let $P=\{a_{i,j}\}_{i\le n,j\le k}$ be a matrix  with the same properties as described in Theorem \ref{theorem:generalized_bvn} with $k\le n$. Then there is a matrix $P'=\{a'_{i,j}\}_{i\le n,j\le n}$ with $\forall i,j: 0 \le a'_{i,j} \le 1$ such that $P = P'\vert_{i\le  n, j\le k}$, 
and $\forall i: \sum_{j=1}^n a'_{i,j} = 1$ and $\forall j: \sum_{i=1}^n a'_{i,j} = 1$. 
\end{lemma}
\begin{proof}
Define $P' = \{a'_{i,j}\}_{i\le n,j\le n}$ as 
\begin{equation}
a'_{i,j} = \left\{
\begin{array}{ll}
a_{i,j} & \text{if }j \le k \\
\frac{1 - \sum_{j'=1}^{k} a_{i,j'}}{n - k} & \text{if } j > k. \\
\end{array}
\right. 
\end{equation}
Then $P'\vert_{i\le n, j\le k} = P$ by definition.  $P'$ satisfies all the requirements from the lemma. A proof of this can be found in the Appendix. 
\end{proof}

\noindent%
By transposing $A$ we can show that the Lemma also holds if $k>n$. 

\subsection{An efficient implementation of the generalized Birkhoff-von Neumann decomposition }\label{sec:efficient_bvn_implementation}

For an implementation of the generalized \acl{BvN} theorem, one can in theory use the proof of Theorem~\ref{theorem:generalized_bvn} and extend the \ac{MRP}-matrix, that we obtained by solving the convex optimization problem from Eq.~\ref{eq:convex_optimization_topk}, to a full $n\times n$-matrix. 
This matrix can then be decomposed into the convex sum of permutation matrices with help of the \acs{BvN} theorem for doubly stochastic matrices after which we can restrict the matrices again to the first $k$ columns.  
Since the complexity of the \acs{BvN} decomposition for square matrices is $\mathcal{O}(n^4\sqrt{n})$ \cite{hopcroft-1973-algorithm, johnson1960algorithm} and hence infeasible for large $n$,  we propose an alternative implementation for $n\times k$ or $k\times n$ matrices with $k<n$, that can be implemented with time complexity $\mathcal{O}(k^3 n^2)$. 

Algorithm~\ref{alg:gen_bvn} gives a structured overview of our algorithm for the generalized \acs{BvN} decomposition.
We start off by noting that the way in which we extended the doubly stochastic matrix from $P$ in the proof of Lemma~\ref{theorem:extension_doubly_stochastic} is not unique.
For any index pair $(i,j), (i',j')$ with $j,j'> k$ we can subtract some value $\beta$ from $a'_{i,j}$ and $a'_{i',j'}$, while adding the same value to $a'_{i',j}$ and $a'_{i,j'}$.  
The resulting matrix will have the same properties as $P'$ and will also be an extension of $P$. 
Therefore, instead of extending $P$ to a full doubly stochastic matrix, we can extend it to an $n\times (k+1)$ matrix $\Ptil$,  where the last column contains the entries that make the values of each row sum to 1.  
In the decomposition we split off matrices that are permutation matrices on the first $k$ columns and have $n-k$ non-zero entries on the last column; see line~\ref{alg1:line2} in Algorithm~\ref{alg:gen_bvn}.

We can use this realization to extend the implementation of the \acs{BvN} algorithm \cite{birkhoff-1940-lattice}, which translates the marginal rank probability matrix into a bipartite graph and uses the Hopcroft-Karp algorithm~\citep{hopcroft-1973-algorithm} to find a perfect matching $m$, which in turn can be translated back into a permutation matrix, $P^m$; see 
line~\ref{alg1:line4}, \ref{alg1:line5} and~\ref{alg1:line6}.\footnote{For the implementation we used  \url{https://networkx.org} and \url{https://github.com/jfinkels/birkhoff}\label{footnote:Hopcroft}}

In the next step, 
line~\ref{alg1:line7},  we calculate the biggest coefficient $\alpha$,  such that subtracting the scaled permutation matrix $\alpha P^m$,  still results in a matrix with only non-negative coefficients.  We add the coefficient-matrix pair to the decomposition and subtract the scaled permutation matrix from $\Ptil$; see
line~\ref{alg1:line8} and~\ref{alg1:line9}.
\begin{algorithm}[t]
\caption{Algorithm for the generalized Birkhoff-von Neumann decomposition.}
\label{alg:gen_bvn}
\begin{algorithmic}[1]
\Require $P\in \text{Mat}_{n\times k}$ with properties as in Theorem \ref{theorem:generalized_bvn} 
\State Initialize $\mathcal{P}=\{\}$ empty decomposition
\State Extend $P$ to $\Ptil$ by adding a column $\{c_i\}_{i=1,\dots, n}$ with values $c_i = 1-\sum_{j=1}^k P_{i,j}$ \label{alg1:line2}
\While{$\Ptil \ne 0$}
	\State Translate $\Ptil$ to a bipartite graph with $n$ resp. $k+1$ vertices on each side with edges between the $i$-th and $j$-th vertex if $P_{i,j}\ne 0$ \label{alg1:line4}
    \State  Find a perfect matching $m$ (with multiplicity of $n-k$ for the last vertex) with the adjusted Hopcroft-Karp algorithm \label{alg1:line5}
    	\State Translate $m$ to a matrix $P^m$, where $P^m\vert_{i\le n, j\le k}$ forms a permutation matrix. \label{alg1:line6}
    	\State $\alpha = \min_{\{i,j \mid P^m_{i,j}\ne 0\}}(\Ptil_{i,j})$ \label{alg1:line7}
    	\State $\mathcal{P} \gets \mathcal{P} + (\alpha, P^m\vert_{i\le n, j\le k})$ \label{alg1:line8}
    \State $\Ptil \gets \Ptil - \alpha P^m$ \label{alg1:line9}
\EndWhile
\State Return $\mathcal{P}$
\end{algorithmic}
\end{algorithm}
By translating the matrix $\Ptil$ into a bipartite graph, where the node corresponding to the $(k+1)$-th column has multiplicity $n-k$,  and adjusting the Hopcroft-Karp algorithm (line~\ref{alg1:line5}) slightly to allow for certain vertices to be matched with higher multiplicity,  we can significantly speed up this part of the algorithm from $n^2\sqrt n $ to $k^2n$.  Since the upper bound of matrices in the decomposition decreases from order $n^2$ to $k n$ the complexity changes as stated in the following Theorem. 
A proof of this statement can be found in the Appendix \ref{sec:proofs_efficient}

\begin{theorem}
Using the modified top-$k$ algorithm for the generalized Birkhoff-von Neumann theorem, Algorithm~\ref{alg:gen_bvn}, a decomposition as described in Theorem~\ref{theorem:generalized_bvn} can be obtained with time complexity $\mathcal{O}(k^3 n^2)$.
\end{theorem}

\subsection{Determining a stochastic policy that avoids rankings with unknown exposure distribution}\label{sec:method_felix}

As explained in Section~\ref{sec:background_outliers}, certain types of rankings can have a non-typical exposure distribution.
Allowing such rankings invalidates the approach by \citet{singh2018fairness}, since a position-based exposure vector $\mathbf{v}$ is used in both the utility calculation and the fairness constraint in their approach.  
In this section our goal is to find a stochastic policy that avoids rankings for which the exposure distribution is unknown.  
We will use a re-sampling strategy, which, after the decomposition step in Eq.~\ref{eq:original_bvn}, rejects rankings with unknown exposure distribution.  
The core idea we present below is based on the assumption that the inter-item dependencies between some of the items is the cause of the shift in exposure and that by shuffling the items between different rankings,  rankings with unknown exposure distribution might be changed into rankings with known exposure distribution.

Algorithm~\ref{alg:FELIX} gives a step-by-step overview of the algorithm used by \method{}. 
Similarly to \citet{wang2020fairness}, we make use of the fact that the Birkhoff-von Neumann decomposition is not unique.  
For most doubly stochastic matrices there is a large number of possible decompositions~\cite{dufosse2018further}, which makes it possible for us to search for a decomposition that does not have a lot of weight on rankings with unknown exposure distribution. 
After determining the MRP matrix $\mathbf{P}$ (line \ref{alg2:line1}), we decompose it into the sum $\mathbf{P} = \sum_{i=1}^M \alpha_i P_{\sigma_i}$. 
In the top-$k$ setting this can be done by using the generalized Birkhoff-von Neumann algorithm (Algorithm~\ref{alg:gen_bvn}); see  Algorithm~\ref{alg:FELIX} line~\ref{alg2:line4}.
We write $\mathcal{P}=\{(\alpha_i, P_{\sigma_i})\}_{i=1,\dots, M}$ for the set of coefficient, matrix pairs in this convex sum. 
Once the matrix is fully decomposed, we divide the resulting coefficient,  permutation matrix pairs $(\alpha_i, P_{\sigma_i})$ into two groups, one containing all the permutations where the corresponding ranking has a known exposure distribution amongst its items and the other one containing pairs corresponding to rankings with unknown exposure distribution:
\begin{align*}
\mathcal{P}_\mathit{known} &= \{(\alpha_i,P_{\sigma_i})\in \mathcal{P} | \sigma_{i}\text{ has known exposure distribution}\} \\ 
\mathcal{P}_\mathit{unknown} &= \mathcal{P} -  \mathcal{P}_\mathit{known}.
\end{align*}
\noindent We  use the elements of $\mathcal{P}_\mathit{known}$ directly as a part of the final decomposition; see lines~\ref{alg2:line5}--\ref{alg2:line7}.
The elements of $\mathcal{P}_\mathit{unknown}$ are aggregated, weighted by their coefficient; see line~\ref{alg2:line8}.
\begin{equation}
\mathbf{\Ptil} = \sum_{(\alpha_i, P_i)\in \mathcal{P}_\mathit{unknown}} \alpha_i \cdot P_i . 
\end{equation}%
Up to scalar multiplication, the resulting matrix ${\mathbf{\Ptil}}$ satisfies the required characteristics of Theorem~\ref{theorem:generalized_bvn} and hence can be decomposed again with the generalized  \acs{BvN} decomposition (Algorithm~\ref{alg:gen_bvn}).

This decomposition-aggregation process repeats for a number of iterations, $\mathit{iter}$ (line~\ref{alg2:line3}--\ref{alg2:line10}).
In each iteration,  the recombination of rankings with unknown exposure distribution makes it possible for the algorithm to group items together that previously have not been together in one ranking.  
Through this re-sampling, the context in which items are presented changes,  which often also means that the exposure distribution of these newly ranked list is known.  
Note that this approach does not remove items from the rankings, but rather shuffles the items among different rankings within the decomposition. 
After $\mathit{iter}$ iterations the remaining rankings with unknown exposure distribution are being added to the policy (line~\ref{alg2:line11}--\ref{alg2:line13})  to ensure the fairness and utility, that was optimized for. 

\begin{algorithm}[t]
\caption{Fairness of Exposure in Light of Incomplete Exposure Estimation (\method)}
\label{alg:FELIX}
\begin{algorithmic}[1]
\Require $\mathcal{D}_q$, $k$,  merit vector $\mathbf{u}$,  position bias vector $\mathbf{v}$,  number of iterations $\mathit{iter}$ 
\State Determine \ac{MRP} matrix $\mathbf{P}$ as in Eq.~\ref{eq:convex_optimization_topk} with $\mathbf{u}$ and $\mathbf{v}$  \label{alg2:line1}
\State Initialize $\pi(\sigma) = 0$, $\forall \sigma \in \mathcal{R}_\mathcal{D}$
\While{$\mathit{iter} \neq 0$} \label{alg2:line3}
    \State  $\mathcal{P} \gets \text{ Decompose } \mathbf{P}$ with Algorithm~\ref{alg:gen_bvn} \label{alg2:line4}
\ForAll{$(\alpha, P_\sigma) \in \mathcal{P}_\mathit{known}$} \label{alg2:line5}
    	\State  $\pi(\sigma) \gets \pi(\sigma) + \alpha$ \label{alg2:line6}
\EndFor \label{alg2:line7}
    \State  $\mathbf{P} \gets \sum_{(\alpha, P_\sigma)\in \mathcal{P}_\mathit{unknown}}  \alpha \cdot P_\sigma $  \label{alg2:line8}
    \State $\mathit{iter} \gets \mathit{iter} - 1$ \label{alg2:line9}
\EndWhile \label{alg2:line10} 
\ForAll{$(\alpha, P_\sigma) \in \mathcal{P}_\mathit{unknown}$} \label{alg2:line11} 
    	\State  $\pi(\sigma) \gets \pi(\sigma) + \alpha$\EndFor \label{alg2:line13} 
\State Return $\pi$
\end{algorithmic}
\end{algorithm}

\subsection{Upshot}
To summarize Section~\ref{sec:method},  we extended the continuous optimization approach to fairness for the top-$k$ setting  in Section~\ref{sec:method_topk}  by proving that the Birkhoff-von Neumann theorem, which is used to decompose the matrix that was attained through the convex optimization, can be extended to a more general setting.  In Section~\ref{sec:efficient_bvn_implementation} we gave an algorithm for the decomposition in the top-$k$ case and discussed an efficient implementation.  This extends the space of use cases to which this approach to fair ranking can be applied.  We will use this in our experiments, which will partly be conducted in the top-$k$ setting.  
In Section~\ref{sec:method_felix} \method{} is introduced, which, by iteratively rejecting rankings with unknown exposure distribution, reduces the probability that such rankings are shown to the user.  

Next, we test the performance of the proposed method for top-$k$ fairness. Furthermore, we investigate how well \method{} is able to avoid rankings with unknown exposure distribution and how this impacts the performance w.r.t.\ fairness and user utility.  

 
\section{Experimental Set-up}
\label{sec:experiments}
We experiment with two variants of our model:  to evaluate our top-$k$ approach to fair ranking we use \method  without re-sampling i.e., with only one iteration, denoted by \textbf{\methodn{1}}; to evaluate our method for reducing the probability of generating rankings with unknown exposure we use $20$ iterations (\textbf{\methodn{20}}).

Our experiments aim to answer the following research questions: 
\begin{enumerate*}[leftmargin=*,label=(RQ\arabic*)]
    \item Can \textbf{\methodn{1}} provide fair top-$k$ rankings while maintaining the user utility compared to the baselines? \label{RQ:fair_top_k_1}
    \item  Can \textbf{\methodn{20}} reduce the probability of showing rankings with unknown exposure distribution to the user without compromising fairness or utility, compared to other methods? \label{RQ:avoid_unknown_exp}
     \label{RQ:no_loss_utility_fairness}
\end{enumerate*}
We use the case of rankings with outliers as an example for rankings with unknown exposure distribution. 
As \citet{sarvi-2021-understanding-arxiv} show, outliers can change the exposure distribution that items collect in expectation; we broadly follow their experimental set-up to be able to compare to prior work that is, for this specific use case, closest to our approach. 

\header{Datasets} 
Our experiments in Section~\ref{sec:results} use two academic search datasets provided by the TREC19 and TREC20  Fair Ranking track.\footnote{\url{https://fair-trec.github.io/}\label{footnote:FRT}}
These datasets come with queries, relevance judgements, and information about the authors and academic articles extracted from the Semantic Scholar Open Corpus.\footnote{\url{http://api.semanticscholar.org/corpus/}} 
See Table~\ref{table:data-statistics} for descriptive statistics of the datasets.
Since we experiment on the task of removing outliers from the top-$k$, which only makes sense for queries with enough items, for testing we only use rankings with at least 20 items. 
The 2020 dataset comes with 200 queries for training and 200 for testing; keeping only the lists with at least 20 papers leaves us with 112 test queries. 
Similarly, the 2019 dataset comes with 631 queries for training and 631 for testing. However the test set contains only 3 queries with more than 20 items, which is not acceptable.  As a pragmatic solution, we keep lists with at least 10 items, which leaves us with 69 test queries, but up-sample each of these queries to 50 items by using the feature vectors of non-relevant items from other random lists as negative samples.

\begin{table}[t]
   \caption{Descriptive statistics of the original and pre-processed TREC Fair Ranking track 2019 and 2020 data. }
   \label{table:data-statistics}
   \centering
   \begin{tabular}{l rr rr}
     \toprule
        &  \multicolumn{2}{c}{2019} &  \multicolumn{2}{c}{2020} \\
       \cmidrule(r){2-3}\cmidrule{4-5}
        &  Train  &  Test &  Train  &  Test \\
      \midrule
      Avg. list size (original) & 4.1 & 4.1 & 23.5& 23.4\\
      Avg. list size (pre-proc.) & 4.1  & 13.0 & 23.5 & 31.9\\
      Avg. \# rel. items/list (original) &2.0 &2.0 &3.7 & 3.4 \\
      Avg. \# rel. items/list (pre-proc.) &2.0 &4.4 & 3.7 & 4.5 \\
      \bottomrule
   \end{tabular}
\end{table}

\header{Experiments} 
We consider approaches where correcting for fairness is a post-processing step. 
We use ListNet~\cite{cao2007learning} as our \ac{LTR} model for the ranking step, with a maximum of 30 epochs, the Adam optimizer with learning rate of 0.02, and early stopping. 
As input to the \ac{LTR} model we use the same data as OMIT\footnote{\url{https://github.com/arezooSarvi/OMIT_Fair_ranking}} with 25 features based on term frequencies, BM25~\cite{robertson2009probabilistic}, and language models~\citep{zhai2001study, tao2006language}.\footnote{Our experimental code is based on  \url{https://github.com/MilkaLichtblau/BA_Laura}.\label{footnote:BA_laura}}

To be able to treat the output of the \ac{LTR} model as the relevance probabilities we normalize the predicted scores to be within the range $[\epsilon, 1]$ with $\epsilon=10^{-4}$. Choosing $\epsilon > 0$ ensures that each item has a non-zero probability of being placed in a ranking. 

As mentioned earlier in this section, we use rankings that contain visible outliers as example for rankings with unknown exposure distribution. Following  \cite{sarvi-2021-understanding-arxiv} we use the number of citations of a paper as a visible feature that may be subject to outliers.
For the context in which outliers are perceived we use the top-$k$ items.  
We use the Z-score with threshold value $2.5$ to determine whether an item can be considered an outlier; see Section~\ref{sec:background_outliers}.

We conduct two types of experiments.  
The first experiment imitates the experimental set-up of \citet{sarvi-2021-understanding-arxiv}, where full rankings are formed but the presence of outliers is only measured in the top-$k$ of each ranking.  The second experiment looks at top-$k$ ranking.  We use $k=10$ in our experiments and aim for individual fairness as opposed to \citep{sarvi-2021-understanding-arxiv,singh2018fairness}, where group fairness is used. 

\header{Baselines} 
To answer research questions \ref{RQ:fair_top_k_1} and \ref{RQ:avoid_unknown_exp},  we compare \textbf{\methodn{1}} and \textbf{\methodn{20}} with the following baselines: 
\begin{description}[leftmargin=\parindent,nosep]
\item[\textbf{\acs{PL}}] As suggested in \citep{diaz2020evaluating}, we use a \ac{PL} ranker initialized with the predicted, normalized scores of the \ac{LTR} model.  
\item[\textbf{\acs{PL}-random}] We use a \ac{PL} ranker over a uniform score distribution as a baseline for a random ranker.
\item[\textbf{Vanilla}] We use the method introduced by \citet{singh2018fairness} with only fairness constraints as the vanilla baseline. This is the model we build upon.
\item[\textbf{Deterministic}] This baseline is ListNet, our traditional \ac{LTR} model.
\item[\textbf{OMIT}] The method introduced in~\citep{sarvi-2021-understanding-arxiv}, where a similar optimization problem is solved as for Vanilla, but with an additional regularizing objective that punishes rankings with a global outlier in the top-$k$. 
\end{description}
For the experiments on the top-$k$, we only sample $k=10$ items from the \acs{PL} models, \textbf{\acs{PL}@10} and \textbf{\acs{PL}-random@10}.  Since \textbf{\methodn{1}} is a novel extension of the Vanilla convex optimization approach for the top-$k$ setting,  we do not have the Vanilla baseline in this setting.  For OMIT we use our top-$k$ convex optimization approach with the additional outlier objective, \textbf{OMIT@10},  to be able to compare the outlier reduction of \methodn{20} and OMIT in the top-$k$ setting.   

\header{Evaluation}
To evaluate fairness we use the EE-L metric \cite{diaz2020evaluating}.  
The target exposure of item $d_i$ is calculated as  $\epsilon^*(d_i) = \epsilon_{total} \cdot u(d_i)/\sum_{j}u(d_j)$, where $\epsilon_{total}$ is the total amount of exposure that users spend in expectation on the ranking, and $u(d_i)$ is the merit, i.e. relevance,  of item $d_i$. Given the expected exposure of all items as a vector $\epsilon$, the \emph{expected exposure loss}, \emph{EE-L} can be calculated as:
\begin{align}
\text{EE-L} = \ell\left(\epsilon, \epsilon^{*}\right) &=\left\|\epsilon-\epsilon^{*}\right\|_{2}^{2}.
\end{align}
Ranking utility performance is measured with NDCG.

For a given query, to evaluate how well a policy $\pi$ performs in avoiding rankings with unknown exposure distribution, we measure the probability that such a ranking is displayed by the policy. 
In our experiments this translates to measuring the probability that a randomly sampled ranking, $\sigma$ contains an outlier: 
\begin{align*}
	P(u \mid \pi) 	
	&=   P(  \sigma \text{ has unknown exposure distribution} \mid  \sigma \sim \pi) \\
	&=_{\text{here}} P( \# \text{ outliers in } \sigma \ge 1 \mid  \sigma \sim \pi).
\end{align*}
Additionally, for comparability with \cite{sarvi-2021-understanding-arxiv}, we measure: 
\begin{align*}
	\operatorname{Outlierness}@k(\pi) &= \mathbb{E}_{\sigma \sim \pi} \sum_{d_i\in \text{top-$k$}(\sigma)} \mathds{1}(\text{$d_i$ is outlier})z(d_i).
\end{align*}
For each metric we report the average value taken over all queries.  
Each experiment was conducted 5 times with different train/vali\-dation split and different random seed. 
Each split uses 80\% of the train-data for training and 20\% of the train-data for validation.  
In our result tables we report the mean results.  
We test for significance with a two tailed paired students t-test,  using the metric values over all queries as input and comparing each method with \methodn{20}.  


\section{Results}
\label{sec:results}
   
Table~\ref{table:top_k_results} and~\ref{table:full_ranking_results} contain the results for our experiments on the top-$k$ and full ranking set-up, respectively.

\begin{table}[!t]
    \caption{Top-$k$ rankings. Significance is measured with a two-tailed paired t-test; all comparisons are against \methodn{20}.
}
    \setlength{\tabcolsep}{2.4pt}
    \centering
    \resizebox{1.0\columnwidth}{!}{%
      \begin{tabular}{@{} l l @{} c @{} cc c cc@{}}
        \toprule
    	 &
	 &
    	\multicolumn{1}{c}{Optimizing} &
        \multicolumn{2}{c}{NDCG{$\uparrow$}} &
        \multicolumn{1}{c}{Fairness{$\downarrow$}} &
        \multicolumn{1}{c}{$P(u\mid \pi)${$\downarrow$}} &
        \multicolumn{1}{c}{Outlierness{$\downarrow$}}
          \\
        \cmidrule(r){4-5}
        \cmidrule(r){6-6}
        \cmidrule(r){7-7}
        \cmidrule{8-8}     
        &
        \multicolumn{1}{c}{Method} &
        \multicolumn{1}{c}{Fairness} &
        \multicolumn{1}{c}{@5} &
        \multicolumn{1}{c}{@10} &
        \multicolumn{1}{c}{EE-L} &
        \multicolumn{1}{c}{@10} &
        \multicolumn{1}{c}{@10} \\
        \midrule
        \multirow{6}{*}{\rotatebox[origin=c]{90}{\textbf{TREC20}}}  
                   & \textbf{\methodn{20}}   & Yes  &  0.203 & 0.279  & \textbf{6.22}  &  \textbf{0.20}  & \textbf{0.115} \\ 
  		  & \textbf{\methodn{1}}   &  Yes  & 0.203 & 0.279 &  6.23 & 0.39\rlap{*} & 0.151\rlap{*}   \\
             \cmidrule(l{1pt}){2-8}
             &   PL@10  &  Yes  & 0.197 & 0.275 & 6.24 & 0.47\rlap{*} & 0.174\rlap{*}  \\
             &   PL-random@10  &  No  & 0.177\rlap{*} & 0.249\rlap{*} & 6.29 & 0.47\rlap{*}  & 0.175\rlap{*} \\
             &   Deterministic  &  No  & \textbf{0.287}\rlap{*} & \textbf{0.370}\rlap{*}  & 7.22\rlap{*} & 0.41\rlap{*} & 0.154\rlap{*} \\
             &   OMIT@10  &  Yes  & 0.198 & 0.273  & 6.34 & 0.33\rlap{*} & 0.132\rlap{*} \\
          \midrule 
         \multirow{6}{*}{\rotatebox[origin=c]{90}{\textbf{TREC19}}}  
             & \textbf{\methodn{20}}   &  Yes  & 0.12 & 0.16  & 5.9 &  \textbf{0.12}  & \textbf{0.08} \\
             & \textbf{\methodn{1}}   &   Yes  &0.12 & 0.16 &  5.9& 0.30\rlap{*}   & 0.12\rlap{*}   \\      
             \cmidrule(l{1pt}){2-8}
             & PL@10  &  Yes  & 0.11 & 0.16 & \textbf{5.8} & 0.35\rlap{*}   & 0.14\rlap{*}    \\
             & PL-random@10  &  No  & 0.10 & 0.15 &  \textbf{5.8} & 0.41\rlap{*}   &  0.16\rlap{*}   \\
             & Deterministic  &  No  & \textbf{0.15} & \textbf{0.21}\rlap{*}   & 7.5\rlap{*}  &  0.25\rlap{*}   &  0.12\rlap{*}  \\
             & OMIT@10  &  Yes  & 0.11 & 0.15  & 6.0 & 0.23\rlap{*}   & 0.10  \\ 
        \bottomrule
      \end{tabular}%
      }\label{table:top_k_results}
    \end{table}

    \begin{table}[t]
    \caption{Full length rankings, remove outliers from the top-$k$. Significance is reported in the same way as in Table~\ref{table:top_k_results}. }
    \setlength{\tabcolsep}{2.4pt}
    \centering
    \resizebox{1.0\columnwidth}{!}{%
      \begin{tabular}{@{} l l @{} c @{} cc c c c @{}}
        \toprule
        &
    	&
    	\multicolumn{1}{c}{Optimizing} &
        \multicolumn{2}{c}{NDCG{$\uparrow$}} &
        \multicolumn{1}{c}{Fairness{$\downarrow$}} &
        \multicolumn{1}{c}{$P(u\mid \pi)${$\downarrow$}} &
        \multicolumn{1}{c}{Outlierness{$\downarrow$}}
          \\
        \cmidrule(r){4-5}
        \cmidrule(r){6-6}
        \cmidrule(r){7-7}
        \cmidrule{8-8}
        &
        \multicolumn{1}{c}{Method} &
        \multicolumn{1}{c}{Fairness} &
        \multicolumn{1}{c}{@5} &
        \multicolumn{1}{c}{@10} &
        \multicolumn{1}{c}{EE-L} &
        \multicolumn{1}{c}{@10} &
        \multicolumn{1}{c}{@10} \\
          \midrule 
            \multirow{6}{*}{\rotatebox[origin=c]{90}{\textbf{TREC20}}}  
             & \textbf{\methodn{20}}   &  \text{Yes} & 0.221 & 0.302 &  \textbf{24.5} &  \textbf{0.24} &  \textbf{0.126}  \\
             \cmidrule(l{1pt}){2-8}
             & Vanilla  &  \text{Yes} &  0.221 & 0.302 &  \textbf{24.5} & 0.40\rlap{*}  & 0.163\rlap{*}  \\
             & PL  &  \text{Yes} & 0.192\rlap{*} & 0.269\rlap{*} &  24.7 & 0.45\rlap{*} &  0.169\rlap{*}  \\
             & PL-random  & \text{No} & 0.178\rlap{*} & 0.249\rlap{*} & 24.9 & 0.47\rlap{*} &  0.175\rlap{*}  \\
             & Deterministic  & \text{No} & \textbf{0.267} & \textbf{0.348} & 24.7 & 0.40\rlap{*} & 0.152 \\
             & OMIT  & \text{Yes} & 0.221 & 0.302 & \textbf{24.5} & 0.34\rlap{*} & 0.139 \\                   
        \midrule
         \multirow{6}{*}{\rotatebox[origin=c]{90}{\textbf{TREC19}}}  
             & \textbf{\methodn{20}} &   \text{Yes}  & 0.15 & 0.22 &  \textbf{46.4} &  \textbf{0.11} & \textbf{0.06} \\
             \cmidrule(l{1pt}){2-8}
             &  Vanilla &  \text{Yes} &  0.16 & 0.22 & \textbf{46.4} & 0.14 & 0.07 \\
             & PL &  \text{Yes}  & 0.12  & 0.17 &  \textbf{46.4} & 0.32\rlap{*}& 0.13\rlap{*} \\
             & PL-random &  \text{No}  & 0.10\rlap{*} & 0.15\rlap{*} & 46.5 & 0.41\rlap{*} & 0.16\rlap{*} \\
             & Deterministic &  \text{No} & \textbf{0.17} & \textbf{0.23} & 46.6 & 0.12 & 0.07  \\
             & OMIT  &  \text{Yes} & 0.13 & 0.18 & 46.5 & 0.15 & \textbf{0.06} \\                        
        \bottomrule
        
      \end{tabular}%
      }\label{table:full_ranking_results}
    \end{table}

\headernodot{\ref{RQ:fair_top_k_1}: 
Can \methodn{1} provide fair top-$k$ rankings while maintaining the user utility compared to the baselines?}
To answer this research question we first compare the performance of \methodn{1} with PL@10, since this is the only baseline that has as its objective to create fair top-$k$ ranking policies. 
For both utility and fairness \methodn{1} performs marginally better on TREC20 data.  In the case of TREC19 data, \methodn{1} still has slightly better user utility; the fairness scores are close to identical.  Overall none of these differences are significant. 

As a sanity check, looking at our other baselines, we see that w.r.t.\ user utility (NDCG), in Table~\ref{table:top_k_results} the deterministic ranker outperforms all probabilistic rankers, which is expected since it is purely optimized for utility. This is reflected in the fairness score, where the deterministic ranker scores significantly worse than \methodn{20}.  W.r.t.\ utility, the random ranker is outperformed by all other probabilistic ranking methods, showing that these methods present users with better results than a uniform ranking policy would.  

To summarize, we find no significant differences in terms of utility or fairness between \methodn{1} on the one hand and the PL-ranker on the one hand.  This makes our approach suitable for top-$k$ ranking under fairness constraints and hence allows us to extend \method{} for this setting.  In the rest of this section, we will see other advantages of \method{} over the PL baseline. 

\headernodot{\ref{RQ:avoid_unknown_exp}: Can \methodn{20} reduce the probability of showing rankings with unknown exposure distribution to the user, without having to compromise fairness or utility, compared to other methods?}
We are interested in the trade-offs between user utility, fairness and the probability of showing rankings with unknown exposure, which is indicated by $P(u\mid \pi)$, in Tables~\ref{table:top_k_results} and~\ref{table:full_ranking_results}.  For the TREC20 data, in both settings \methodn{20} successfully improves $P(u\mid \pi)$ while maintaining the NDCG@10 and EE-L scores compared to all baselines.  
Our main baseline to compare with for this research question is OMIT, as it is the only model that optimizes for presenting fewer outliers in the top-$k$ positions. 
Compared to OMIT, \methodn{20} achieves  significantly better results in terms of $P(u\mid \pi)$ for both settings, while keeping the same (or better) scores for other metrics.  For the top-$k$ experiment, we also see a significant improvement w.r.t.\ $P(u\mid \pi)$, compared to \methodn{1}: iteratively re-sampling successfully reduces the number of rankings with unknown exposure distribution in the policy. 
For the TREC19 data we can still observe that \methodn{20} offers the best trade-off between the three objectives in the top-$k$ setting.  However, the improvements w.r.t.\ the outlier removal are less significant in the full length experiments.  Since for this dataset we used an up-sampling strategy that adds varying negative samples, the variation within these experiments is much higher, which makes the results less reliable and causes the observed differences to be less significant.  Still, since the results broadly agree with the results for the more reliable TREC20 dataset, we take this as confirmation for the conclusions drawn there. 

We also report the Outlierness metric, as introduced in~\citep{sarvi-2021-understanding-arxiv}, to show that the improvement of \methodn{20} is not just due to the evaluation metric introduced in this paper but that there is an actual improvement w.r.t.\ the outlier use case.

We conclude that in our experiments,  \methodn{20} is able to effectively reduce the probability that a ranking with unknown exposure distribution is shown to the user,  without a drop in utility or fairness,  compared to other fair ranking methods and OMIT.  

\header{Discussion}
If we compare our results to those in~\citep{sarvi-2021-understanding-arxiv}, OMIT does not perform as well as expected w.r.t.\ $P(u\mid \pi)$ and Outlierness. 
We see two reasons for this.  
First,  OMIT considers outliers in the context of the whole list, while we consider outliers in the context of the top-$k$ that they are presented in; their approach is able to remove outliers defined in the global context from the rankings but does not consider the outliers in the local context they are presented in, which is what we are evaluating for. 

Second,  in this paper we consider individual fairness, while~\citet{sarvi-2021-understanding-arxiv} report results on group fairness.
For individual fairness the number of constraints is much higher, therefore the space we are optimizing over is  smaller, making it challenging for OMIT to find a good solution that is optimized for both utility and reducing outliers while satisfying all the fairness constraints. 
\methodn{20} does not suffer from this, since, instead of adding an additional objective term to the optimization, it intervenes at the decomposition step, making it independent from the constraints used in the optimization.  

This comparison shows that \method{} is very general in terms of use cases that it can be applied to.  
The condition that determines whether a ranking has a known exposure distribution can be focused on each individual ranking without having to rely on global assumptions. 
This allows us to really consider inter-item dependencies, while OMIT needs to work with the heuristic of global outliers instead.
This also highlights the advantages of \method{} over the PL-ranker method. While for most experiments there was no significant difference in utility and fairness between those two methods,  considering inter-item dependencies within the rankings is not possible for the PL approach to fair ranking. 


\section{Sensitivity Analysis of \protect\method{}} 
\label{sec:analysing}
 \label{sec:analaysis_results}
 
Given the results obtained in the previous section, we now analyze the ability of \method{} to reduce the number of rankings with unknown exposure distribution along two important dimensions:
\begin{enumerate*}[leftmargin=*,label=(D\arabic*)]
\item the number of available item candidates; and \label{D:num_candidates}
\item the number of re-sampling iterations, $iter$ (see line \ref{alg2:line3} in Algorithm \ref{alg:FELIX}). \label{D:resampling_iterations}
\end{enumerate*} 

For the TREC datasets most queries have less than 40 items, hence, we use a simulated set-up. 
This gives us more control, allowing us to observe \method{}'s behaviour for different distributions and numbers of candidate items. 
Each analysis is conducted with a series of $m=100$ simulated sets of $n$ items (one can think of these item-sets as corresponding to $m$ imaginary queries). 
Since we want to focus on the effectiveness of \method, rather than the quality of the predicted labels,  we assume that for each item we know the correct probability that an item is relevant to users.  
For our analysis we sample these scores uniformly in the interval $[0,1]$.
The feature that is used for the outlier detection is sampled from a different probability distribution. We conduct experiments on the uniform, normal, log-normal, and power-law distribution to see how dependent the results are on the underlying data distribution.
Each of the different distributions has a different base probability for a list of a given length to contain an outlier, and hence can be seen as different levels of difficulty for removing the rankings with unknown exposure distribution.  With the definition of outliers used in this paper and a list length of 10, the probability that such a list contains an outlier is 0.6\% for the uniform,  2.7\% for the normal, 36.3\% for the log normal and 60.5\% for the power-law distribution.


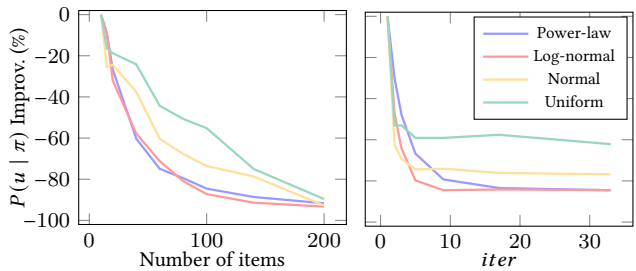
\begin{figure}
\centering

\newcommand{\SpacingX}{0.1em}
\newcommand{\Width}{0.20\textwidth}
\newcommand{\Height}{0.16\textwidth}
\newcommand{\BarWidth}{0.001\textwidth}
\newcommand{\BarOffset}{0.00\textwidth}
	\begin{tikzpicture}[font=\small]
    		\begin{axis}[
            	xlabel = Number of items ,
            	width=\Width,
            	height=\Height,
            	y label style={yshift=-1.4em},
            	x label style={yshift=0.9em},
            	bar width=\BarWidth,
			scale only axis,
			ylabel={$P(u\mid \pi)$ Improv.  (\%) },
			ymin=-102,
			ymax=2, 	
        	]       
        	\addplot[line width=0.3mm,color=blue!40!white] table[x=num_candidates, y=powerlaw,col sep=comma]{figures/candidates.csv}; 
        	\addplot[line width=0.3mm,color=red!40!white] table[x=num_candidates, y=log_normal,col sep=comma]{figures/candidates.csv}; 
        	\addplot[line width=0.3mm,color=orange!60!yellow!40!white] table[x=num_candidates, y=gaussian,col sep=comma]{figures/candidates.csv};
        	\addplot[line width=0.3mm,color=green!60!blue!40!white] table[x=num_candidates, y=uniform,col sep=comma]{figures/candidates.csv}; 
    		\end{axis}
	\end{tikzpicture}%
	\begin{tikzpicture}[font=\small]
    		\begin{axis}[
            	xlabel = $iter$,
            	width=\Width,
            	height=\Height,
            	x label style={yshift=0.9em},
            	bar width=\BarWidth,
			scale only axis,
			ymin=-102,
			ymax=2, 
			yticklabels=none, 
			legend style={nodes={scale=0.8, transform shape}},
        	]       
        	\addplot[line width=0.3mm,color=blue!40!white] table[x=number_of_resamples, y=powerlaw,col sep=comma]{figures/resample_experiment.csv}; 
        	\label{pgfplots:plot1}
        	\addlegendentry{Power-law}
        	\addplot[line width=0.3mm,color=red!40!white] table[x=number_of_resamples, y=log_normal,col sep=comma]{figures/resample_experiment.csv}; 
        	\label{pgfplots:plot2}    		
        	\addlegendentry{Log-normal}
        	\addplot[line width=0.3mm,color=orange!60!yellow!40!white] table[x=number_of_resamples, y=gaussian,col sep=comma]{figures/resample_experiment.csv};
    		\addlegendentry{Normal}
        	\label{pgfplots:plot3}
        	\addplot[line width=0.3mm,color=green!60!blue!40!white] table[x=number_of_resamples, y=uniform,col sep=comma]{figures/resample_experiment.csv}; 
        	\label{pgfplots:plot4}
    		\addlegendentry{Uniform}
    		\end{axis}
    		legend style={at={(0.9,0.5)},anchor=west}
	\end{tikzpicture}%
        \caption{Sensitivity analysis. Relative reduction in~$P(u\mid \pi)$ in $\%$ on the y-axis for different numbers of available candidate items (left) and different numbers of iterations (right). }
        \label{fig:analysis}
\end{figure}

\header{\ref{D:num_candidates} Candidate items}
The left plot in Fig.~\ref{fig:analysis} shows the relative reduction of rankings with outliers with a varying number of candidate items.  We use 20 re-sampling iterations. 
We see that for all distributions,  \method performs increasingly better as the number of items increases.
Having more items to shuffle between various rankings gives the method more flexibility in putting outlier items into different contexts, in which they do not appear as outliers.

\header{\ref{D:resampling_iterations} Re-sampling parameter}
The right plot of Fig.~\ref{fig:analysis} shows how well \method is able to remove outliers from the rankings based on the number of re-sampling iterations,  which is the only new hyper-parameter introduced by our method.  We use 100 candidate items per query. 
We find that with an increasing number of re-samples, \method can remove more outliers. Nevertheless, the gains seem to be diminishing, depending on the distribution after 5--20 iterations.  

\header{Broader implications}
Ranking systems often work in two stages, where in the first stage a certain number of documents are retrieved and in the second stage they are re-ranked with help of a learning to rank method. Our analysis of the number of candidate items \ref{D:num_candidates} can help deciding on how many items to retrieve in the first stage.  Moreover, the analysis of the re-sampling parameter \ref{D:resampling_iterations}  can help with deciding on a good performance/computation time trade-off when choosing the number of allowed re-sampling iterations.  


\section{Conclusion}

Motivated by recent work on the impact of outliers on the exposure distribution within a ranking, we introduced the task of fair ranking under incomplete exposure estimation.  
We defined a new method, \method{}, that avoids showing rankings to the user which, due to inter-item dependencies, have unknown exposure distribution.  
We extended the convex optimization approach to fairness to the top-$k$ setting and gave an efficient implementation of the algorithm that makes it feasible, even for a large number of items.  We  showed empirically that \method{} is able to significantly reduce the probability of generating  rankings with unknown exposure, without hurting user utility or fairness compared to previous fair ranking methods.   

\method{} is a first step towards fair ranking in cases where due to inter-item dependencies there is uncertainty about the exposure distribution of some rankings.  By defining an efficient algorithm for the top-$k$ setting, we enable the usage of the convex optimization approach towards fairness for use cases with a large number of items, which previously had been infeasible.  We discussed that this approach gives more flexibility than other methods and allows, for example, to consider the relationship between items. 

One limitation of our work is that, since the policy achieved by the convex optimization is only fair in expectation, this approach is most useful for head queries with a large number of repetitions. 
Use cases where this might be applied include job search, where next to the individual fairness criterion a correction for historical biases should be considered, or item search for items that are frequently bought.  
Second, our results are based on the assumption that the unknown exposure comes from inter-item dependencies and that the same items that cause one ranking to have unknown exposure distribution,  when placed in another context will result in a ranking with known exposure distribution.  This assumption holds for rankings with visible outliers, however, to prove the generalizability of this approach,  experiments with other use cases are needed. 
Lastly, 
to have enough flexibility within the Birkhoff-von Neumann decomposition algorithm, enough entries of this matrix need to be non-zero.  Using group fairness with only two groups, results in a marginal rank probability matrix that is a linear combination of just two permutation matrices~\citep{singh2018fairness}.  More groups introduce more stochasticity, therefore this method is particularly interesting when working with individual fairness or a larger number of groups. 
 
A potential direction for future work is to investigate whether \method{} can be extended for different user models. In this work we assume that most rankings follow a position-based exposure distribution. For other user-models like the cascade model a different approach might be necessary.  Also, more research needs to be done on inter-item dependencies between items in a ranking and their impact on the exposure for different use cases. Phenomena like outliers or click bait have been explored to some extent but other types of cognitive bias that impact how we perceive items in relation to others have been broadly unexplored in the context of ranking systems. Lastly, extending user models to include inter-item dependencies such as outliers might allow for a more direct approach to fair ranking in cases where the exposure distribution is unknown.

\section*{Data and Code}
To facilitate reproducibility of our work, all code and parameters are shared at \url{https://github.com/MariaHeuss/2022-SIGIR-FOE-Incomplete-Exposure}.

\section*{Acknowledgements}
We thank our reviewers for valuable feedback.
This research was supported by the Hybrid Intelligence Center,  a 10-year program funded by the Dutch Ministry of Education, Culture and Science through the Netherlands Organisation for Scientific Research, \url{https://hybrid-intelligence-centre.nl}, and by Ahold Delhaize. 
All content represents the opinion of the authors, which is not necessarily shared or endorsed by their respective employers and/or sponsors.

\appendix

\section{Proofs}

\subsection{Extended proof for the generalized \acl{BvN}} \label{sec:proofs_gen_bvn}

We give a more detailed proof of  Lemma \ref{theorem:extension_doubly_stochastic} and Theorem \ref{theorem:generalized_bvn}. Recall that we say that 
$P'\in \mathbb{R}^{n'\times k'}$ is an \emph{extension} of $P\in \mathbb{R}^{n\times k}$ if $n'\ge n, k' \ge k$, and $P_{i,j}=P'_{i,j}$ for all $(i,j)$ with $i\le n$ and $j\le k$. We denote this by $P'\vert_{i\le n, j\le k} = P$. 

\begin{lemma}
\label{theorem:long_extension_doubly_stochastic}
Let $P=\{a_{i,j}\}_{i\le n,j\le k}$ be a matrix  with the same properties as described in Theorem \ref{theorem:generalized_bvn} with $k\le n$. Then there is a matrix $P'=\{a'_{i,j}\}_{i\le n,j\le n}$ with $\forall i,j: 0 \le a'_{i,j} \le 1$ such that $P = P'\vert_{i\le  n, j\le k}$, 
and $\forall i: \sum_{j=1}^n a'_{i,j} = 1$ and $\forall j: \sum_{i=1}^n a'_{i,j} = 1$. 
\end{lemma}
\begin{proof}
Define $P' = \{a'_{i,j}\}_{i\le n,j\le n}$ as 
\begin{equation}
a'_{i,j} = \left\{
\begin{array}{ll}
a_{i,j} & \text{if }j \le k \\
\frac{1 - \sum_{j'=1}^{k} a_{i,j'}}{n - k} & \text{if } j > k. \\
\end{array}
\right. 
\end{equation}
Then $P'\vert_{i\le n, j\le k} = P$ by definition.  
Since for all $i$,  $0 \le \sum_{j=1}^k a_{i,j} \le 1$ we also have $0 \le \frac{1 - \sum_{j'=1}^{k} a_{i,j'}}{n - k} \le 1$. Moreover, for all $i \le n $:
\begin{align*}
    \sum_{j=1}^n a'_{i,j} & = \sum_{j=1}^k a_{i,j} + \sum_{j=k+1}^n \frac{1 - \sum_{j'=1}^{} a_{i,j'}}{n - k} \\
    &= \sum_{j=1}^k a_{i,j} + (n-k) \cdot  \frac{1 - \sum_{j'=1}^{k} a_{i,j'}}{n - k} \\
    &= \sum_{j=1}^k a_{i,j} + (1 - \sum_{j'=1}^{k} a_{i,j'})\\
    &= 1,
\end{align*}
where we used in the second equality that we sum over $(n-k)$ times the same value. We know that the columns of the matrix sum to 1 for all $j\le k$, since this is the case for $P$. For $j>k$ we have: 
\begin{align*}
    \sum_{i=1}^n a'_{i,j} & = \frac{1}{n-k} (\sum_{j=k}^n\sum_{i=1}^n a'_{i,j})\\
    & = \frac{1}{n-k}(n-\sum_{j=1}^k\sum_{i=1}^n a'_{i,j})\\ 
    & = \frac{n-k}{n-k} = 1.
\end{align*}
Here in the first equality we used that all columns from the $k$-th column are the same. In the second equality we used that since all rows are summing to $1$, the sum of all rows (and therefore also the sum of all columns) equals $n$. The last equality simply uses the fact that each of the first $k$ columns sums to $1$. 
\end{proof}

\noindent%
We use this Lemma to prove the generalized \acl{BvN} theorem. Let $k\le n$. 

\begin{theorem}
\label{theorem:extended_generalized_bvn}
Any matrix $P=\{a_{i,j}\}_{i\le n,j\le k}$
with $\forall i,j: 0 \le a_{i,j} \le 1$,  $\forall j: \sum_{i=1}^n a_{i,j} = 1$ and   $\forall i: \sum_{j=1}^k a_{i,j} \le 1$ can be written as the convex sum $P=\sum_{l=1}^m \alpha_l \cdot P_l$ of permutation matrices $P_l$ with coefficients $\alpha_l \in [0,1]$ such that $\sum_{l=1}^m \alpha_l=1$.  
\end{theorem}

\begin{proof}
In Lemma \ref{theorem:long_extension_doubly_stochastic} we show that $P$ can be extended to a doubly stochastic matrix $P'$, i.e. $P= P'\vert_{i\le n, j\le k}$. 
For this matrix $P'$, the theorem by Birkhoff and von Neumann
states that we can find a decomposition into the convex sum of permutation matrices, 
$P' = \sum_{l=1}^m \alpha_l P'_l$,
with $\alpha_l \in [0,1]$,  $\sum_{l=1}^m \alpha_l=1$ and  $P'_l$  permutation matrices.  This induces a decomposition of the original matrix $P$: 
\begin{align*}
P = \sum_{l=1}^m \alpha_l P'_l\vert_{i\le n, j\le k}.\qquad \qedhere
\end{align*}
\end{proof}

\subsection{Complexity of the generalized \acl{BvN} algorithm}\label{sec:proofs_efficient}

In this section we prove the following claim from Section~\ref{sec:efficient_bvn_implementation}: 

\begin{theorem}
Using the modified top-$k$ algorithm for the generalized Birkhoff-von Neumann theorem, Algorithm~\ref{alg:gen_bvn}, a decomposition as described in Theorem~\ref{theorem:generalized_bvn} can be obtained with time complexity $\mathcal{O}(k^3 n^2)$.
\end{theorem}

\begin{proof}
The time complexity of Algorithm \ref{alg:gen_bvn} depends on the complexity of the adjusted Hopcroft-Karp algorithm (line \ref{alg1:line5}) and the number of times it needs to be executed (line \ref{alg1:line4}--\ref{alg1:line9}), which is equal to the number of permutation matrices in the decomposition.  \citet{hopcroft-1973-algorithm} show that the time complexity of the Hopcroft-Karp algorithm is $\mathcal{O}((m + l)\sqrt{l})$, where $l$ is the number of vertices and 
$m$ is the number of edges in the biparate graph.  
For the baseline approach we have $l = 2n$ and $m=n^2$,  therefore the complexity of the Hopcroft-Karp algorithm in this setting would be $\mathcal{O}(n^2\sqrt(n))$.  Using our approach instead, we have $l = n+(k+1)$ and $m=n\cdot (k+1)$ which reduces the complexity to $\mathcal{O}(kn\sqrt(n))$. Furthermore since the maximum length of each augmenting path is bounded by $2\cdot k$,  we can substitute the  $\sqrt{n}$ term with $k$ (see Corollary 2 and Theorem 3 of \cite{hopcroft-1973-algorithm}). This gives us a time complexity of $\mathcal{O}(k^2n)$ for the full matching algorithm. 
For the number of matrices in the decomposition,  \citet{johnson1960algorithm} define an upper bound of  $n^2-2n+2$ permutation matrices, which means that the total complexity of the Birkhoff-von Neumann algorithm equals $\mathcal{O}(n^4\sqrt{n})$.  Since for our  algorithm, a loose upper bound for the number of permutation matrices is $k\cdot n$, the algorithm proposed in this paper has a time complexity of only $\mathcal{O}(n^2k^3)$, which makes it much more feasible than the more naive algorithm proposed in Section \ref{sec:method_topk} for large values of $n$. 
\end{proof}

\clearpage

\bibliographystyle{ACM-Reference-Format}
\balance
\bibliography{references}
\clearpage 

\end{document}